\documentclass[groupedaddress, superscriptaddress, citesort]{revtex4}
\usepackage{color}
\usepackage{bm}
\usepackage{graphicx}
\usepackage{amsbsy}
\usepackage{amsmath}
\usepackage{amsfonts}
\usepackage{amsthm}
\usepackage{amssymb}
\usepackage{tikz}
\usepackage{makecell}
\usepackage{diagbox}

\usepackage{graphicx}
\usepackage{sidecap}

\usepackage{verbatim}

\usepackage{setspace}

\usepackage{enumitem}
\setlist{nolistsep}

\usepackage{quoting}
\quotingsetup{vskip=0pt}

\begin{document}

\theoremstyle{plain}
\newtheorem{theorem}{Theorem}
\newtheorem{lemma}[theorem]{Lemma}
\newtheorem{corollary}[theorem]{Corollary}
\newtheorem{conjecture}[theorem]{Conjecture}
\newtheorem{proposition}[theorem]{Proposition}

\theoremstyle{definition}
\newtheorem{definition}{Definition}

\theoremstyle{remark}
\newtheorem*{remark}{Remark}
\newtheorem{example}{Example}

\def\be{\begin{equation}}
\def\ee{\end{equation}}
\def\ba{\begin{align}}
\def\ea{\end{align}}

\newcommand{\mE}{\mathcal{E}}
\newcommand{\mU}{\mathcal{U}}
\newcommand{\mA}{\mathcal{A}}
\newcommand{\mF}{\mathcal{F}}
\newcommand{\mI}{\mathcal{I}}
\newcommand{\mH}{\mathcal{H}}
\newcommand{\mL}{\mathcal{L}}
\newcommand{\mM}{\mathcal{M}}
\newcommand{\mT}{\mathcal{T}}
\newcommand{\mN}{\mathcal{N}}

\newcommand{\fm}{\mathcal{F}_{\bf{m}}}
\newcommand{\am}{\mathcal{A}^{\textbf{m}}}
\newcommand{\dm}{\mathcal{D}(\mathrm{H}_{{\bf m}})}
\newcommand{\lr}{\rangle\langle}
\newcommand{\la}{\langle}
\newcommand{\ra}{\rangle}
\newcommand{\tr}{{\rm Tr}}

\newcommand{\mc}[1]{\mathcal{#1}}
\newcommand{\mbf}[1]{\mathbf{#1}}
\newcommand{\mbb}[1]{\mathbb{#1}}
\newcommand{\mrm}[1]{\mathrm{#1}}

\newcommand{\bra}[1]{\langle #1|}
\newcommand{\ket}[1]{|#1\rangle}
\newcommand{\braket}[3]{\langle #1|#2|#3\rangle}
\newcommand{\ip}[2]{\langle #1|#2\rangle}
\newcommand{\op}[2]{|#1\rangle \langle #2|}

\newcommand{\mbN}{\mathbb{N}}

\definecolor{eric}{rgb}{0,.5,.2}
\newcommand{\eric}[1]{{\color{eric} #1}}

\newcommand{\review}[1]{{\color{red} #1}}
\topmargin=-15mm\oddsidemargin=-2mm\textwidth=164mm\textheight=240mm
\def\ra{\rangle}
\def\la{\langle}
\def\ot{\otimes}
\def\be{\begin{equation}}
\def\ee{\end{equation}}
\def\ba{\begin{array}}
\def\ea{\end{array}}
\def\t{\tilde}
\def\Nb{{I\!\! N}}
\def\Rb{{I\!\! R}}
\def\Zb{{Z\!\!\! Z}}
\def\Fb{{I\!\! F}}
\def\Cb{{\Bbb C}}
\def\cb{{\Bbb C}}

\baselineskip=18pt

\title {Quantum coherence quantifiers based on the R\'{e}nyi $\alpha$-relative entropy}
\author{Lian-He Shao}
\affiliation{College of Computer Science, Shaanxi Normal University, Xi'an 710062,
China}
\author{Yongming Li}
\email{liyongm@snnu.edu.cn}
\affiliation{College of Computer Science, Shaanxi Normal University, Xi'an 710062,
China}
\author{Yu Luo}
\affiliation{College of Computer Science, Shaanxi Normal University, Xi'an 710062,
China}
\author{Zhengjun Xi}
\affiliation{College of Computer Science, Shaanxi Normal University, Xi'an 710062,
China}

\begin{abstract}
The resource theories of quantum coherence attract a lot of attention in recent years. Especially, the monotonicity property plays a crucial role here. In this paper we investigate the monotonicity property for the coherence measures induced by the R\'{e}nyi $\alpha$-relative entropy which present in [Phys. Rev. A 94, 052336, 2016]. We show that the  R\'{e}nyi $\alpha$-relative entropy of coherence does not in general satisfy the monotonicity requirement under the subselection of measurements condition and it also does not satisfy the extension of monotonicity requirement which presents in [Phys. Rev. A 93, 032136, 2016]. Due to the R\'{e}nyi $\alpha$-relative entropy of coherence can act as a coherence monotone quantifier, we examine the trade-off relations between coherence and mixedness. Finally, some properties for the single qubit of R\'{e}nyi $2$-relative entropy of coherence are derived.
\end{abstract}

\maketitle

\section{Introduction}

Coherence arising from quantum superposition rule, is an important resources in quantum information theory. Coherence is discussed in the interference phenomena and it's know due to the role of phase coherence in optical phenomena~\cite{Mandel1}. A rigorous framework for quantifying coherence was proposed by Baumgratz $et.al.$ and they proposed several measures of coherence which are based on information distance measures including relative entropy and $l_1$ norm~\cite{Plenio1}. The quantification framework of quantum coherence stimulated many further considerations which include other coherence measures~\cite{Rana1,Streltsov1,Shao1}, the operational interpretations of quantum coherence~\cite{Napoli1,Singh2,Winter1}, the relationship between quantum entanglement, quantum discord and quantum deficit~\cite{Peng2,Xi1,Chitambar1,Yao1,Ma1}, quantification of coherence in infinite dimensional system~\cite{Zhang1,Xu1}, the other properties are similar to quantum entanglement theory~\cite{Liu15,Fugang1,Chen1,Chitambar3,Radhakrishnan1,Mani1,Peng1,haijun1,Yu1,Bu1,Wang1,Piani1,Bromley1,Du1}

From the view point of the definition, one can straightforwardly quantify the coherence in a given basis by measuring the distance between the quantum state $\rho$ and its nearest incoherent state. Baumgratz \emph{et al.} give four necessary criteria~\cite{Plenio1} which any quantity should fulfill them. Given a finite-dimensional Hilbert space $\mathcal{H}$ with $d=dim(\mathcal{H})$. We note that $\mathcal{I}$ is the set of quantum states which is called incoherent state that are diagonal in a fixed basis $\{|i\rangle\}_{i=1}^d$,$\{K_n\}$ is a set of Kraus operators,and satisfies $\sum_n K_n^{\dagger}K_n= \mathbb{I}$ with $K_n \mathcal{I} K_n^\dagger\subset \mathcal{I}$. Then any proper measure of the coherence $C$ must satisfy the following conditions:

(C1) $C(\rho)\geq 0$ for all quantum states $\rho$, and $C(\rho)=0$ if and only if $\rho\in \mathcal{I}$.

(C2a) Monotonicity under all the incoherent completely positive and trace preserving (ICPTP) maps $\Phi$: $C(\rho)\geq C(\Phi(\rho)) $, where $\Phi(\rho)=\sum_n K_n \rho K_n^\dag$.

(C2b) Monotonicity for average coherence under subselection based on measurements outcomes: $C(\rho)\geq \sum_n p_n C(\rho_n) $, where $\rho_n=\frac{K_n\rho K_n^\dag}{p_n}$ and $p_n=\mathrm{Tr}(K_n \rho K_n^\dag)$.

(C3) Non-increasing under mixing of quantum states: $\sum_n p_n C(\rho_n)\geq C(\sum_n p_n \rho_n )$ for any ensemble $\{p_n, \rho_n\}$.

The R\'{e}nyi entropy is important in quantum information theory. It can be used as a measure of entanglement~\cite{Wang4}. In Ref.~\cite{Hiai1}, Mosonyi and Hiai define the R\'{e}nyi $\alpha-$relative entropy, which can act as an information distance measure. In~\cite{Chitambar4}, Chitambar \emph{et al.} propose that the R\'{e}nyi $\alpha-$relative entropy of coherence fulfills condition C1 and C2a for $\alpha\in [0,2]$, then we call the R\'{e}nyi $\alpha-$relative entropy of coherence is a coherence  monotone~\cite{Streltsov2}. As we know, the condition C2b is important as it allows for sub-selection based on measurement outcomes, a process available in well controlled quantum experiments and it's also difficult to verify~\cite{Plenio1}. A natural question arises immediately, is the condition C2b satisfied for the R\'{e}nyi $\alpha-$relative entropy of coherence?

In this paper, we will resolve the above question. In Sect. 2, we review basic points for the R\'{e}nyi $\alpha-$relative entropy of coherence
. In Sect. 3, We prove that the R\'{e}nyi $\alpha-$relative entropy of coherence does't fulfill the condition C2b and it also does not fulfill the extension condition C2b presented in~\cite{Rastegin1}. We give the tradeoff relation between the R\'{e}nyi $\alpha-$relative entropy of coherence and mixedness in Sect. 4. The case of the R\'{e}nyi $\alpha-$relative entropy of coherence for a single qubit is discussed in Sect 5. In Sect. 6 we give the summary of results.

\section{The R\'{e}nyi $\alpha-$relative entropy of coherence}\label{II}
In this section, we recall basic points of the R\'{e}nyi $\alpha-$relative entropy of coherence present in~\cite{Chitambar4}. For $\alpha\in[0,\infty]$ the R\'{e}nyi $\alpha$- relative entropy of the states $\rho$ by $\delta$ is defined by~\cite{Hiai1}
\begin{eqnarray}
S_{\alpha}(\rho\|\delta):=\frac{1}{\alpha-1}\log\tr(\rho^\alpha\delta^{1-\alpha})\;.
\end{eqnarray}
This quantity is contractive for all $\alpha\in[0,2]$. Since the R\'{e}nyi $\alpha-$relative entropy can act as an information distance measure~\cite{Hiai1}. Then, we can define the R\'{e}nyi $\alpha-$relative entropy of coherence as:
\begin{eqnarray}
C_\alpha(\rho):=\min_{\delta\in\mI} S_{\alpha}(\rho\|\delta).
\end{eqnarray}
Note that in the limit $\alpha\to 1$, $S_{\alpha\rightarrow 1}(\rho\| \delta)$ gives the relative entropy $S(\rho\|\delta )=\tr(\rho\log\rho)-\tr(\rho\log\delta)$. Let $\delta=\sum_i q_i |i\lr i|$ be some incoherent states, then the analytical expression of $C_\alpha(\rho)$ can be obtained as~\cite{Chitambar4}
\begin{eqnarray}
C_\alpha(\rho):=\min_{\{q_i\}} \frac{1}{\alpha-1}\log\sum_{i} q_{i}^{1-\alpha}\langle i|\rho^\alpha|i\rangle.
\end{eqnarray}
Eq.(3) can be further simplified as~\cite{Chitambar4}
\begin{eqnarray}
C_{\alpha}(\rho)=\frac{\alpha}{\alpha-1}\log \sum_{i}\left(\langle i|\rho^\alpha|i\rangle\right)^{1/\alpha}.
\end{eqnarray}
In this paper, we don't consider the cases for $\alpha=0$ and the limit $\alpha\to 1$. For $\alpha=0$, the R\'{e}nyi relative entropy of coherence is always equal to $0$. For the limit $\alpha\to 1$, a detailed study for the standard relative entropy of coherence is presented in~\cite{Plenio1}

\section{The monotoncity property}\label{III}
First we show that $C_\alpha (\rho)$ fulfills the condition C3 for $\alpha\in [0,1)$. In Ref.~\cite{Hiai1}, it is shown that $S_\alpha(\rho || \delta)$ is convexity for $\alpha\in [0,1)$. For any ensemble$\{p_i,\rho_i\}$, we assume the incoherent states $\delta_i^*$ are closet with respect to $\rho_i$, then we have
\begin{eqnarray}
C_\alpha(\sum_i p_i \rho_i)&=&\min_{\delta\in \mI} S_\alpha(\sum_i p_i \rho_i ||\delta ) \nonumber \\
&\leq & S_\alpha(\sum_i p_i \rho_i ||\sum_i p_i \delta_i^* ) \nonumber \\
&\leq & \sum_i p_i S_\alpha( \rho_i || \delta_i^* ) =\sum_i p_i C_\alpha( \rho_i),
\end{eqnarray}
where the second inequality using the convexity of $S_\alpha(\rho || \delta)$. We conclude that for $\alpha\in [0,1)$, $C_\alpha (\rho)$ cannot increase under mixing of quantum states, then $C_\alpha (\rho)$ fulfills the condition C3 for $\alpha\in [0,1)$. $S_\alpha(\rho || \delta)$ isn't convexity anymore for $\alpha\in (1,2]$~\cite{Hiai1}, there may exist some cases which lead to $C_\alpha(\sum_i p_i \rho_i) > \sum_i p_i C_\alpha( \rho_i)$.
 The condition C3 combined with C2b, implies C2a. In Ref.~\cite{Chitambar4}, Chitambar \emph{et al.} study the R\'{e}nyi $\alpha-$relative entropy of coherence fulfills C2a for different kinds of incoherent operations. They don't consider whether the R\'{e}nyi $\alpha-$relative entropy of coherence fulfills the condition C2b. This motivates us to study whether C2b is satisfied for $C_\alpha (\rho)$.

Now we use the example which presented in~\cite{Plenio1} to show that condition C2b is violated. We choose $|0\rangle=\begin{pmatrix} 1  \\ 0  \\ 0 \end{pmatrix}$,$|1\rangle=\begin{pmatrix} 0  \\ 1  \\ 0 \end{pmatrix}$ and $|2\rangle=\begin{pmatrix} 0  \\ 0  \\ 1 \end{pmatrix}$ are the prescribed
orthonormal basis. The two Kraus operators are written as
\begin{equation}
K_{1}=
\begin{pmatrix}
0 & 1 & 0 \\
0 & 0 & 0 \\
0 & 0 & a
\end{pmatrix}
, \qquad
K_{2}=
\begin{pmatrix}
1 & 0 & 0 \\
0 & 0 & b \\
0 & 0 & 0
\end{pmatrix},
\end{equation}
where the complex numbers $a$ and $b$ obey $|a|^{2}+|b|^{2}=1$. This condition guarantees that  $\sum_n K_n^{\dagger}K_n= \mathbb{I}$.
The density matrix is presented as
\begin{equation}
\rho=
\frac{1}{4}
\begin{pmatrix}
1 & 0 & 1 \\
0 & 2 & 0 \\
1 & 0 & 1
\end{pmatrix}
.
\end{equation}
After applying this channel to the density matrix $\rho$, we obtain the output states:
\begin{align}
&\rho_{1}=
\frac{1}{2+|a|^{2}}
\begin{pmatrix}
2 & 0 & 0 \\
0 & 0 & 0 \\
0 & 0 & |a|^{2}
\end{pmatrix}
, \label{php120}\\
&\rho_{2}=
\frac{1}{1+|b|^{2}}
\begin{pmatrix}
1 & b^{*} & 0 \\
b & |b|^{2} & 0 \\
0 & 0 & 0
\end{pmatrix}
.
\end{align}
With the probabilities:
\begin{equation}
p_{1}=\frac{2+|a|^{2}}{4}
\ , \qquad
p_{2}=\frac{1+|b|^{2}}{4}
\ . \label{prb12}
\end{equation}
By using Eq(4), we obtain the R\'{e}nyi $\alpha-$relative entropy of coherence for $\rho$ as
\begin{equation}
C_\alpha(\rho)=\frac{\alpha}{\alpha-1}\log [\frac{1}{2}+(\frac{1}{2})^{\frac{1}{\alpha}}].
\end{equation}
Note that the operator $K_1$ makes $C_\alpha(\rho_1) = 0$, we only need to calculate $C_\alpha(\rho_2)$.
\begin{equation}
C_\alpha(\rho_2)=\frac{\alpha}{\alpha-1}\log [(\frac{1}{1+|b|^2})^\frac{1}{\alpha}+(\frac{|b|^2}{1+|b|^2})^\frac{1}{\alpha}].
\end{equation}
We choose $b=1$, substituting it into Eqs.(9) and (10), we then get
\begin{equation}
p_2C_\alpha(\rho_2)=\frac{1}{2}\frac{\alpha}{\alpha-1}\log [(\frac{1}{2})^\frac{1}{\alpha}+(\frac{1}{2})^\frac{1}{\alpha}].
\end{equation}
Using the inequality $x+y\geq 2\sqrt{xy}$, we obtain
\begin{equation}
\frac{1}{2}+(\frac{1}{2})^\frac{1}{\alpha}\geq 2\sqrt{\frac{1}{2}\times(\frac{1}{2})^\frac{1}{\alpha}}=\sqrt{2\times(\frac{1}{2})^\frac{1}{\alpha}}.
\end{equation}
The equality holding if and only if $\alpha=1$, Thus, for $\alpha\in (0,1)$
\begin{equation}
C_\alpha(\rho)=\frac{\alpha}{\alpha-1}\log [\frac{1}{2}+(\frac{1}{2})^{\frac{1}{\alpha}}] < \frac{1}{2}\frac{\alpha}{\alpha-1}\log [(\frac{1}{2})^\frac{1}{\alpha}+(\frac{1}{2})^\frac{1}{\alpha}]=p_2C_\alpha(\rho_2).
\end{equation}
If we choose $b=\frac{1}{2}$,  substituting it into Eqs.(9) and (10), we get
\begin{equation}
p_2C_\alpha(\rho_2)=\frac{3}{8}\frac{\alpha}{\alpha-1}\log [(\frac{2}{3})^\frac{1}{\alpha}+(\frac{1}{3})^\frac{1}{\alpha}].
\end{equation} 
\begin{figure}\label{fig1}
  \includegraphics[scale=0.7]{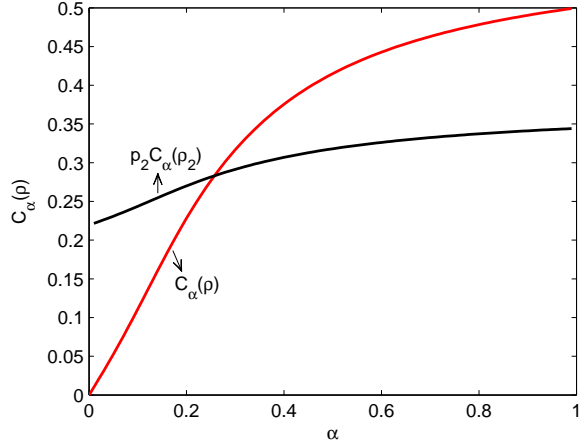}
\caption{  Comparison between $C_{\alpha}(\rho)$ and $p_2C_{\alpha}(\rho_2)$ for $b=\frac{1}{2}$. The black line shows $p_2C_{\alpha}(\rho_2)$. The red line shows  $C_{\alpha}(\rho)$ .}
\label{Fig_1}
\end{figure}
We plot $C_{\alpha}(\rho)$ and $p_2C_\alpha(\rho_2)$ in Fig.1. From the Fig.1, we can also find some cases to illustrate $C_\alpha(\rho)<\sum_n p_n C_\alpha(\rho_n)$.

Recently, Yu $et.al.$ propose an alternative framework for quantifying coherence which is more flexible and convenient for applications than the original one~\cite{Yu060302}. Their framework can be expressed as follows. Any proper measure of the coherence $C$ must satisfy the following three conditions:

(B1) Nonnegativity:  $C(\rho)\geq 0$ for all quantum states $\rho$, and $C(\rho)=0$ if and only if $\rho\in \mathcal{I}$.

(B2) Monotonicity:  $C(\rho)\geq C(\Phi(\rho)) $, where $\Phi(\rho)=\sum_n K_n \rho K_n^\dag$ is an incoherent operation.

(B3) Additivity of coherence for subspace-independent states: $C(p_1 \rho_1\oplus p_2 \rho_2)=p_1C(\rho_1)+p_2C(\rho_2)$ for block-diagonal states $\rho$ in the incoherent basis, where density operators $\rho_1$ and $\rho_2$ are defined on the two independent subspaces, $p_1$ and $p_2$ are two possibility coefficients with $p_1+p_2=1$ and $p_1 \rho_1\oplus p_2 \rho_2=\begin{pmatrix}p_1\rho_1 & 0  \\0 & p_2 \rho_2  \end{pmatrix} $.

The above three conditions(B1,B2,B3) are fulfilled by all the coherence measures based on the original four conditions(C1,C2a, C2b C3). Thus, this framework provides us an alternative method to illustrate that the measure of coherence induced by R\'{e}nyi $\alpha$-relative entropy must violate C2b. We consider a state $\rho=p_1 \rho_1\oplus p_2 \rho_2$, with $\rho_1=\frac{1}{2}(|0\rangle+|1\rangle)(\langle0|+\langle1|)$ and $\rho_2=\frac{1}{3}(|2\rangle+|3\rangle+|4\rangle)(\langle2|+\langle3|+\langle4|)$~\cite{Yu060302}. We choose the computational basis ${|i\rangle}_{i=0}^4$ as the reference basis, then we have
\begin{eqnarray}
C_{\alpha}(\rho_1)&=&1, \nonumber \\
C_{\alpha}(\rho_2)&=&\log 3, \nonumber \\
C_{\alpha}(\rho)&=&\frac{\alpha}{\alpha-1}\log [(\frac{1}{2})^{1/{\alpha}}+\frac{3}{2}(\frac{1}{3})^{1/{\alpha}}].
\end{eqnarray}
\begin{figure}\label{fig1}
  \includegraphics[scale=0.7]{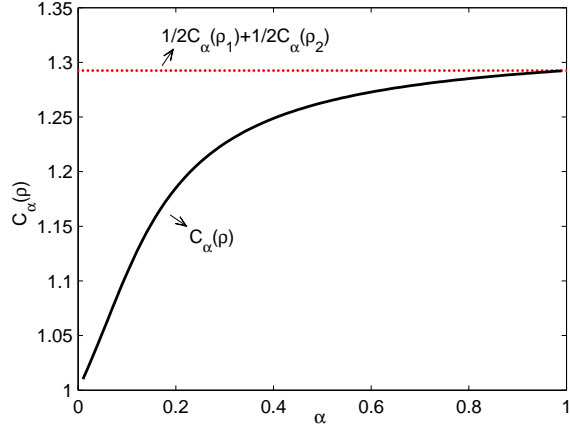}
\caption{Comparison between $C_{\alpha}(\rho)$ and $\frac{1}{2}C(\rho_1)+\frac{1}{2}C(\rho_2)$. The black solid line shows $p_2C_{\alpha}(\rho_2)$. The red dotted line shows  $\frac{1}{2}C(\rho_1)+\frac{1}{2}C(\rho_2)$ .}
\label{Fig_1}
\end{figure}
 We plot $C_{\alpha}(\rho)$ and $\frac{1}{2}C(\rho_1)+\frac{1}{2}C(\rho_2)$ in Fig.2. It is shown that 
\begin{eqnarray}
C_{\alpha}(\rho)&=&C_{\alpha}(p_1 \rho_1\oplus p_2 \rho_2)\neq p_1C_{\alpha}(\rho_1)+p_2C_{\alpha}(\rho_2).
\end{eqnarray} 
We note that for the limit $\alpha\to 1$, the R\'{e}nyi $\alpha-$relative entropy of coherence will become the standard relative entropy of coherence, thus we have $C_{\alpha\to 1}(\rho)= p_1C_{\alpha\to 1}(\rho_1)+p_2C_{\alpha\to 1}(\rho_2)$.
Therefore, the R\'{e}nyi $\alpha-$relative entropy of coherence must violate C2b in general.

From the above examples, we then conclude that condition C2b, i.e., $C_\alpha(\rho)\geq \sum_n p_n C_\alpha(\rho_n)$ is not generally true for the
measure of coherence induced by R\'{e}nyi $\alpha$-relative entropy.

In Ref.~\cite{Rastegin1}, Rastegin studies the Tsallis relative $\alpha$ entropies of coherence $C_\alpha^T(\rho)=\min \frac{1}{\alpha-1}[\tr(\rho^\alpha\sigma^{1-\alpha})-1]$ and give an extension of condition C2b. The extension of condition C2b can be represented as
\begin{equation}
\sum_n p_n^\alpha q_n^{1-\alpha} C_\alpha^T(\rho)\leq C_\alpha^T(\rho).
\end{equation}
where $p_n=\tr (K_n \rho K_n^\dag)$,$q_n=\tr (K_n \sigma K_n^\dag)$.

We compare the Tsallis relative $\alpha$ entropies of coherence with the R\'{e}nyi $\alpha$-relative entropy of coherence $C_\alpha(\rho)=\min \frac{1}{\alpha-1}\log\tr(\rho^\alpha\sigma^{1-\alpha})$,  $C_\alpha^T(\rho)$ and $C_\alpha(\rho)$ are different with the $\log$ function. We recall that the fidelity of coherence $C_F(\rho)=1-\max_{\delta\in \mI}\sqrt{F(\rho,\delta)}$, fulfills conditions C1, C2a and C3, it violates condition C2b~\cite{Shao1}. Another related quantity is geometric coherence $C_g(\rho)=1-\max_{\delta\in \mI} F(\rho,\delta)$, fulfills conditions C1, C2a C2b and C3~\cite{Streltsov1}. Although $C_F(\rho)$ and $C_g(\rho)$ are different with the square root function, but $C_g(\rho)$ fulfills the conditions C2b and $C_F(\rho)$ doesn't fulfill condition C2b. Another question arises immediately, is the extension of condition C2b satisfied for the R\'{e}nyi $\alpha-$relative entropy of coherence $C_\alpha(\rho)$?

We also use the above example to solve this problem. According to Ref.~\cite{Chitambar4}, when we use the R\'{e}nyi $\alpha$-relative entropy to quantify coherence, the optimal incoherent state for $\rho$ is
\begin{equation}
\delta =
\frac{1}{2+\sqrt{2}}
\begin{pmatrix}
1 & 0 & 0 \\
0 & \sqrt{2} & 0 \\
0 & 0 & 1
\end{pmatrix}
.
\end{equation}
With the corresponding probabilities:
\begin{eqnarray}
&&q_{1}=\tr(K_{1}\delta K_{1}^{\dagger})=\frac{\sqrt{2}+|a|^{2}}{2+\sqrt{2}} ,  \nonumber \\
&&q_{2}=\tr(K_{2}\delta K_{2}^{\dagger})=\frac{1+|b|^{2}}{2+\sqrt{2}} .
\end{eqnarray}
We also choose $b=1$, substituting it into Eqs.(10),(12) and (21), we can obtain the expression of $p_2^{\alpha}q_2^{1-\alpha}C_\alpha(\rho_2)$
\begin{eqnarray}
p_2^{\alpha}q_2^{1-\alpha}C_\alpha(\rho_2)&=&
(\frac{1+|b|^2}{4})^\alpha (\frac{1+|b|^2}{2+\sqrt{2}})^{1-\alpha}\frac{\alpha}{\alpha-1}
\log [(\frac{1}{1+|b|^2})^\frac{1}{\alpha}+(\frac{|b|^2}{1+|b|^2})^\frac{1}{\alpha}] \nonumber \\
&=& (\frac{1}{2})^\alpha (\frac{2}{2+\sqrt{2}})^{1-\alpha}\frac{\alpha}{\alpha-1}
\log [(\frac{1}{2})^\frac{1}{\alpha}+(\frac{1}{2})^\frac{1}{\alpha}] .
\end{eqnarray}
Compare Eq.(13) with Eq.(22), after some simple algebraic operation, we obtain
\begin{eqnarray}
(\frac{1}{2})^\alpha\times(\frac{2}{2+\sqrt{2}})^{1-\alpha} \geq \frac{1}{2}.
\end{eqnarray}
The equality holding if and only if $\alpha=1$. Thus, for $\alpha\in (0,1)$
\begin{eqnarray}
C_\alpha(\rho) < p_2C_\alpha(\rho_2)< p_2^{\alpha}q_2^{1-\alpha}C_\alpha(\rho_2).
\end{eqnarray}
From the above example, we then conclude that the extension of condition C2b, i.e., $\sum_n p_n^\alpha q_n^{1-\alpha} C_\alpha(\rho)\leq C_\alpha(\rho)$ is not generally true for the
measure of coherence induced by R\'{e}nyi $\alpha$-relative entropy.

\section{The R\'{e}nyi $\alpha-$relative entropy of coherence and mixedness }\label{IIII}
In order to be a meaningful resource quantum quantifier for coherence, the minimal requirements are the conditions C1 and C2a for any quantity $C$~\cite{Streltsov2}. We have proved that the R\'{e}nyi $\alpha-$relative entropy of coherence does not fulfill the condition C2b and the extension of condition C2b, but the R\'{e}nyi $\alpha-$relative entropy of coherence is satisfied
the minimal requirements to be a coherence quantifier for $\alpha\in [0,2]$, thus the R\'{e}nyi $\alpha-$relative entropy of coherence can act as a coherence monotone quantifier~\cite{Streltsov2}.  An important problem for quantifying coherence is the relationship between quantum coherence quantities and mixedness. The trade-off between some quantities and mixedness have been discussed in~\cite{Singh1,Cheng1}. Here we focus on the trade-off between the R\'{e}nyi $\alpha-$relative entropy of coherence and mixedness.
\begin{theorem}
For $0<\alpha\leq 2 $ and $\alpha\neq1$ , the upper bound of the R\'{e}nyi $\alpha-$relative entropy of coherence is given by
\begin{equation}
C_\alpha(\rho)\leq \log d+\log \tr(\rho^2),
\end{equation}
and the trade-off between $C_\alpha(\rho)$ and mixedness can express as
\begin{equation}
\frac{\ln2}{d-1}C_\alpha(\rho)+M(\rho)\leq 1,
\end{equation}
where $M(\rho):=\frac{d}{d-1}[1-\tr(\rho^2)]$
\end{theorem}
\begin{proof}
We only prove the case of $\alpha\in(0,1)$, $\alpha\in(1,2]$ is completely analogous. For $\alpha\in(0,2]$ and $\alpha\neq1$, we choose $\delta$ is the completely mixed state $\delta=\sum_i\frac{1}{d}|i \rangle \langle i|$, then we can obtain
\begin{eqnarray}
C_\alpha(\rho)&=&\min_{\delta\in I} S_\alpha(\rho||\delta)  \nonumber \\
&\leq&S_\alpha(\rho||\sum_i\frac{1}{d}|i\rangle\langle i|) \nonumber \\
&=&\frac{1}{\alpha-1}\log \tr[\rho^\alpha(\sum_i\frac{1}{d}|i \rangle \langle i|)^{1-\alpha}] \nonumber \\
&=& \frac{1}{\alpha-1}\log \tr (d^{\alpha-1} \rho^\alpha) \nonumber \\
&=& \frac{1}{\alpha-1}(\log d^{\alpha-1}+\log \tr \rho^\alpha).
\end{eqnarray}
The inequation holding is that the completely mixed state $\sum_i\frac{1}{d}|i\rangle\langle i|$ may not be the optimal incoherent state for $\rho$.
According to~\cite{Rastegin1}, for $\alpha\in(0,2]$ and $\alpha\neq1$, the function $\varepsilon\rightarrow \frac{\varepsilon^{\alpha-1}}{\alpha-1}$ is concave, applying Jensen's inequality, we then have
 \begin{eqnarray}
\frac{\tr(\rho^\alpha)}{\alpha-1}=\sum_i \lambda_i\frac{\lambda_i^{\alpha-1}}{\alpha-1}\leq \frac{[\tr(\rho^2)]^{\alpha-1}}{\alpha-1}.
\end{eqnarray}
Here $\lambda_i$ are the  eigenvalues of $\rho$ and obey the normalization condition. For $\alpha\in(0,1)$, $\alpha-1<0$, then $ \tr(\rho^\alpha)\geq [\tr(\rho^2)]^{\alpha-1}$, so $\frac{\log \tr(\rho^\alpha)}{\alpha-1}\leq \frac{\log [\tr(\rho^2)]^{\alpha-1}}{\alpha-1}$, then
\begin{eqnarray}
C_\alpha(\rho)&\leq& \frac{1}{\alpha-1}(\log d^{\alpha-1}+\log \tr \rho^\alpha) \nonumber \\
&\leq& \frac{\log d^{\alpha-1}+\log [\tr(\rho^2)]^{\alpha-1}}{\alpha-1} \nonumber \\
&=& \log d +\log \tr(\rho^2).
\end{eqnarray}
Using $\ln x\leq x-1$, we then have
\begin{eqnarray}
C_\alpha(\rho)&\leq& \log d +\log \tr(\rho^2) \nonumber \\
&=&\log d \tr(\rho^2) \nonumber \\
&\leq&\frac{d \tr\rho^2-1}{\ln 2}.
\end{eqnarray}
Combine $M(\rho):=\frac{d}{d-1}[1-\tr(\rho^2)]$ and the above inequation, we get
\begin{eqnarray}
\frac{\ln2}{d-1}C_\alpha(\rho)+M(\rho)\leq1.
\end{eqnarray}
\end{proof}
Eq.(25) provide an upper bound on R\'{e}nyi $\alpha-$relative entropy of coherence in terms of the purity $\tr(\rho^2)$. Eq.(26) showed that when  mixedness increases, an upper bound on R\'{e}nyi $\alpha-$relative entropy of coherence decreases.

\section{the R\'{e}nyi $\alpha-$relative entropy of coherence for a single qubit }\label{IIIII}
Due to the analytical expression of $C_\alpha(\rho)$ for qubit is complicated, so we consider a simple case that we choose $\alpha=2$ for the coherence quantity. The qubit states can write as
\begin{equation}
\rho =
\begin{pmatrix}
 a & b^*  \\
b & 1-a
\end{pmatrix}
.
\end{equation}
The eigenvalues of $\rho$ are expressed as
\begin{eqnarray}
&&\lambda_1=\frac{1}{2}+\frac{1}{2}\sqrt{1+4|b|^2+4a^2-4a}  \nonumber \\
&&\lambda_2=\frac{1}{2}-\frac{1}{2}\sqrt{1+4|b|^2+4a^2-4a}  .
\end{eqnarray}
Combine Eq(4) and Eq(32), we can obtain,
\begin{equation}
C_2(\rho)=2\log(\sqrt{a^2+|b|^2}+\sqrt{|b|^2+(1-a)^2}).
\end{equation}
Due to $0\leq \lambda_{1,2}\leq 1$, thus we have
\begin{equation}
|b|^2\leq a(1-a).
\end{equation}
For the given $a$, the minimum of $C_2(\rho)$ is zero. The maximum of $C_2(\rho)$ can be expressed as
\begin{equation}
C_2^{max}(\rho)=2\log(\sqrt{a}+\sqrt{1-a}).
\end{equation}
Now, we consider the precise trade-off between $C_2(\rho)$ and mixedness for qubit case.

\begin{theorem}
For the given $a$, the trade-off between $C_2(\rho)$ and mixedness $M(\rho)$ can express as
\begin{equation}
\ln2C_2(\rho)+M(\rho)< 2\sqrt{a(1-a)}\leq1.
\end{equation}
where $M(\rho):=\frac{d}{d-1}[1-\tr(\rho^2)]$
\end{theorem}

\begin{proof}
Combine Eq.(31) and Eq.(34), we can obtain
\begin{eqnarray}
\ln 2 C_2(\rho)+M(\rho)&=&\ln2 \log(\sqrt{a^2+|b|^2}+\sqrt{|b|^2+(1-a)^2})^2+M(\rho)  \nonumber \\
&=&\ln(\sqrt{a^2+|b|^2}+\sqrt{|b|^2+(1-a)^2})^2+M(\rho) \nonumber \\
&\leq&(\sqrt{a^2+|b|^2}+\sqrt{|b|^2+(1-a)^2})^2-1+2[1-(a^2+|b|^2)-(|b|^2+(1-a)^2)]      \nonumber \\
&=& 1-(\sqrt{a^2+|b|^2}-\sqrt{|b|^2+(1-a)^2})^2,
\end{eqnarray}
where the inequation uses $\ln x\leq x-1$. Next, we illustrate that $1-(\sqrt{a^2+|b|^2}-\sqrt{|b|^2+(1-a)^2})^2$ is a monotone increasing function of $|b|^2$ for the given $a$ and $|b|^2\leq a(1-a)$. We take derivative with respect to $|b|^2$, then we have
 \begin{eqnarray}
 \frac{d (1-(\sqrt{a^2+|b|^2}-\sqrt{|b|^2+(1-a)^2})^2)}{d (|b|^2)}&=&-2(\sqrt{a^2+|b|^2}-\sqrt{|b|^2+(1-a)^2})(\frac{1}{\sqrt{a^2+|b|^2}}-\frac{1}{\sqrt{|b|^2+(1-a)^2}}) \nonumber \\
 &=&\frac{2(\sqrt{a^2+|b|^2}-\sqrt{|b|^2+(1-a)^2})^2}{\sqrt{a^2+|b|^2}\sqrt{|b|^2+(1-a)^2}}\geq0.
\end{eqnarray}
 The maximum value of $1-(\sqrt{a^2+|b|^2}-\sqrt{|b|^2+(1-a)^2})^2$ is obtained when $|b|^2= a(1-a)$. Thus,
\begin{eqnarray}
\ln 2 C_2(\rho)+M(\rho)&\leq& 1-(\sqrt{a^2+|b|^2}-\sqrt{|b|^2+(1-a)^2})^2  \nonumber \\
&\leq& 1-(\sqrt{a}-\sqrt{1-a})^2 \nonumber \\
&=&2\sqrt{a(1-a)}.
\end{eqnarray}
The equality holding if and only if $\ln(\sqrt{a^2+|b|^2}+\sqrt{|b|^2+(1-a)^2})^2=(\sqrt{a^2+|b|^2}+\sqrt{|b|^2+(1-a)^2})^2+1$ and $|b|^2= a(1-a)$. After some simple algebraic operation, we can get $a=0$, $b=1$ or $a=1$, $b=0$. Those two solutions can't be represented quantum states. Then we can obtain
\begin{equation}
\ln2C_2(\rho)+M(\rho)< 2\sqrt{a(1-a)}.
\end{equation}
\end{proof}
\begin{figure}\label{fig1}
  \includegraphics[scale=0.7]{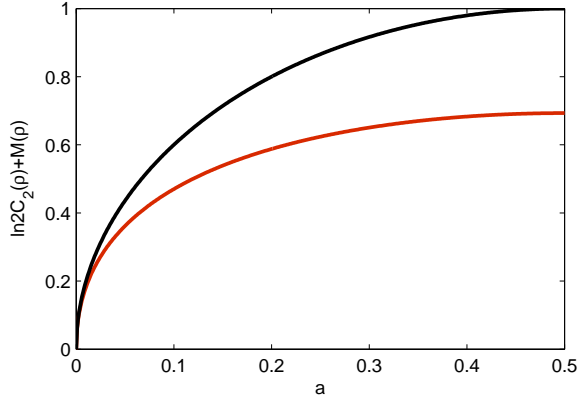}
\caption{ The values of $\ln2C_2(\rho)+M(\rho)$ and $2\sqrt{a(1-a)}$ as functions of $a$}
\label{Fig_1}
\end{figure}
In Fig.3, we plot the left-hand side of Eq.(41) by a red line and the right-hand side of Eq.(41) by a black line. From the Fig.3, we can see $\ln2C_2(\rho)+M(\rho)$ is smaller than $2\sqrt{a(1-a)}$. It can ensure the correctness of the Theorem 2.

\section{\bf summary}\label{IIIIII}
In this paper, we show that the  R\'{e}nyi $\alpha-$relative entropy of coherence does not satisfy condition C2b and extension of C2b for $\alpha\in(0,1)$  by presenting examples. Thus the measure of coherence induced by the  R\'{e}nyi $\alpha-$relative entropy can't be called coherence measure~\cite{Streltsov2}. Due to the R\'{e}nyi $\alpha$-relative entropy of coherence fulfill the condition C1 and C2a, so the R\'{e}nyi $\alpha$-relative entropy of coherence can be called as a coherence monotone quantifier. The R\'{e}nyi $\alpha$-relative entropy of coherence fulfills the minimal requirements to be a meaningful resource quantum quantifier for coherence~\cite{Streltsov2}, then we examine the trade-off relations between coherence and mixedness. Some properties are further exemplified with a single qubit for $\alpha=2$. Our findings complement the results present in Ref.~\cite{Chitambar4}.

\bigskip
\noindent {\bf Acknowledgments}  This work was supported by NSFC (Grant Nos. 11271237, 11671244, 61671280 and 11401361)and the Higher School Doctoral Subject
Foundation of Ministry of Education of China (Grant No. 20130202110001)and Fundamental Research Funds for the Central Universities(GK201502004 and 2016CBY003).


\begin{thebibliography}{18}
\bibitem{Mandel1} L. Mandel and E. Wolf, Optical Coherence and Quantum Optics(Cambridge University Press, Cambridge, Uk,1995).
\bibitem{Plenio1} T.Baumgratz, M. Cramer and M. B. Plenio, Quantifying coherence, Phys. Rev. Lett. \textbf{113}, 140401(2014).
\bibitem{Rana1}S. Rana, P. Parashar and M. Lewenstein, Trace-distance measure of coherence, Phys. Rev. A \textbf{93} 012110 (2016).
\bibitem{Streltsov1}A. Streltsov, U. Singh, H. S. Dhar, M. N. Bera and G. Adesso, Measuring Quantum Coherence with Entanglement, Phys. Rev. Lett. \textbf{115} 020403 (2015).
\bibitem{Shao1} L.-H. Shao, Z. Xi, H. Fan and Y. Li, Fidelity and trace-norm distances for quantifying coherence, Phys. Rev. A \textbf{91} 042120 (2015).
\bibitem{Napoli1} C. Napoli, T. R. Bromley, M. Cianciaruso, M. Piani, N. Johnston and G. Adesso, Robustness of Coherence: An Operational and Observable Measure of Quantum, Phys. Rev. Lett. \textbf{116} 150502 (2016).
\bibitem{Singh2} U. Singh, M. N. Bera, A. Misra and A. K. Pati, Erasing Quantum Coherence: An Operational Approach, arXiv:1506.08186.
\bibitem{Winter1} A. Winter and D. Yang, Operational Resource Theory of Coherence, Phys. Rev. Lett. \textbf{116} 120404 (2016).
\bibitem{Peng2} Y. Peng, Y.-R. Zhang, Z.-Y. Fan, S. Liu and H. Fan, Complementary relation of quantum coherence and quantum correlations in multiple measurements, arXiv:1608.07950.
\bibitem{Xi1} Z. Xi, Y. Li and H. Fan, Quantum coherence and correlations in quantum system, Scientific Reports \textbf{5} 10922 (2015).
\bibitem{Chitambar1} E. Chitambar and Min-Hsiu Hsieh, Relating the Resource Theories of Entanglement and Quantum Coherence, Phys. Rev. Lett. \textbf{117} 020402 (2016).
\bibitem{Yao1} Y. Yao, X. Xiao, L. Ge and C. P. Sun, Quantum coherence in multipartite systems, Phys. Rev. A \textbf{92} 022112 (2015).
\bibitem{Ma1} J. Ma, B. Yadin, D. Girolami, V. Vedral and M. Gu, Converting Coherence to Quantum Correlations, Phys. Rev. Lett. \textbf{116} 160407 (2016).
\bibitem{Zhang1}Y.-R. Zhang, L.-H. Shao, Y. Li and H. Fan, Quantifying coherence in infinite-dimensional systems,Phys. Rev. A \textbf{93} 012334 (2016).
\bibitem{Xu1} J. Xu, Quantifying coherence of Gaussian states, Phys. Rev. A \textbf{93} 032111 (2016).
\bibitem{Liu15} C. L. Liu, X.-D. Yu, G. F. Xu and D. M. Tong, Ordering states with coherence measures, Quantum Inf. Process 15, 4189(2016)
\bibitem{Fugang1} F.-G. Zhang, L.-H. Shao, Y. Luo and Y. Li, Ordering states of Tsallis relative $\alpha$-entropies of coherence, Quantum Inf Process 16, 31(2017). 
\bibitem{Chen1}J.-J. Chen, J. Cui, Y.-R. Zhang and H. Fan, Coherence susceptibility as a probe of quantum phase transitions, Phys. Rev. A \textbf{94} 022112 (2016).
\bibitem{Chitambar3} E. Chitambar, A. Streltsov, S. Rana, M. N. Bera, G. Adesso and M. Lewenstein, Assisted Distillation of Quantum Coherence, Phys. Rev. Lett. \textbf{116} 070402 (2016).
\bibitem{Radhakrishnan1} C. Radhakrishnan, M. Parthasarathy, S. Jambulingam and T. Byrnes, Distribution of Quantum Coherence in Multipartite Systems, Phys. Rev. Lett. \textbf{116} 150504 (2016).
\bibitem{Mani1} A. Mani and V. Karimipour, Cohering and decohering power of quantum channels, Phys. Rev. A \textbf{92} 032331 (2015).
\bibitem{Peng1}Y. Peng, Y. Jiang and H. Fan, Maximally coherent states and coherence-preserving operations, Phys. Rev. A \textbf{93} 032326 (2016).
\bibitem{haijun1} H.-J. Zhang, B. Chen, M. Li, S.-M. Fei and G.-L. Long, Estimation on Geometric Measure of Quantum Coherence, Commun. Theor. Phys. \textbf{67} (2017) 166-170.
\bibitem{Yu1} X.-D. Yu, D.-J. Zhang, C. L. Liu and D. M. Tong, Measure-independent freezing of quantum coherence, Phys. Rev. A \textbf{93} 060303 (2016).
\bibitem{Bu1} K. Bu, U. Singh and J. Wu, Catalytic coherence transformations, Phys. Rev. A \textbf{93} 042326 (2016).
\bibitem{Wang1} J. Wang, Z. Tian, J. Jing and H. Fan, Irreversible degradation of quantum coherence under relativistic motion, Phys. Rev. A \textbf{93} 062105 (2016).
\bibitem{Piani1}M. Piani, M. Cianciaruso, T. R. Bromley, C. Napoli, N. Johnston and G. Adesso, Robustness of asymmetry and coherence of quantum states Phys. Rev. A \textbf{93} 042107 (2016)
\bibitem{Bromley1} T. R. Bromley, M. Cianciaruso and G. Adesso, Frozen Quantum Coherence, Phys. Rev. Lett. \textbf{114} 210401 (2015).
\bibitem{Du1} S. Du, Z. Bai, and Y.Guo, Conditions for coherence transformations under incoherent operations, Phys. Rev. A \textbf{91} 052120 (2015).
\bibitem{Wang4}  Y.-X. Wang, L.-Z. Mu, V. Vedral and H. Fan, Entanglement R\'{e}nyi $\alpha$ entropy, Phys. Rev. A \textbf{93} 022324
\bibitem{Hiai1} M. Mosonyi and F. Hiai, On the Quantum R\'{e}nyi Relative Entropies and Related Capacity Formulas, IEEE Trans. Inf. Theory, \textbf{57} 2474 (2011)
\bibitem{Chitambar4} E. Chitambar and G. Gour, Comparison of incoherent operations and measures of coherence, Phys. Rev. A \textbf{94}  052336 (2016).
\bibitem{Yu060302}X.-D. Yu, D.-J. Zhang, G. F. Xu and D. M. Tong, Alternative framework for quantifying coherence, Phys. Rev. A \textbf{94}, 060302 (2016)
\bibitem{Streltsov2} A. Streltsov, G. Adesso and M. B. Plenio, Quantum Coherence as a Resource, arXiv:1609.02439 (2016).
\bibitem{Rastegin1} A. E. Rastegin, Quantum-coherence quantifiers based on the Tsallis relative $\alpha$ entropies, Phys. Rev. A \textbf{93}  032136 (2016).
\bibitem{Hiai2} F. Hiai, M. Mosonyi, D. Petz and C. Beny, Quantum f-divergences and error correction, Rev. Math. Phys. \textbf{23} 691 (2011).
\bibitem{Singh1} U. Singh, M. N. Bera, H. S. Dhar and A. K. Pati, Maximally coherent mixed states: Complementarity between maximal coherence and mixedness, Phys. Rev. A \textbf{91} 052115 (2015).
\bibitem{Cheng1} S. Cheng and M. J. W. Hall, Complementarity relations for quantum coherence, Phys. Rev. A \textbf{92} 042101 (2015).
\end{thebibliography}
\end{document}